\documentclass[article, ij4uq]{ij4uq}

\usepackage[hang]{footmisc}
\fancypagestyle{plain}{%
  \fancyhf{}
  \fancyhead[R]{
  }
  \fancyfoot[R]{\small\bf\thepage }
  \fancyfoot[L]{
  }
  }

\renewcommand{\dmyy}{}
\renewcommand{\myyear}{2026}
\renewcommand{\today}{}

\usepackage{mathtools}
\usepackage{siunitx}
\usepackage{url}

\usepackage{hyperref}
\usepackage{cleveref}
\crefname{exam}{example}{examples}
\Crefname{exam}{Example}{Examples}
\crefname{enumi}{step}{steps}
\Crefname{enumi}{Step}{Steps}
\crefname{equation}{}{}

\theoremstyle{remark}
\newtheorem*{continuancex}{Example \continuanceref  { }Continued}

\makeatletter
\newenvironment{continuance}[1]{%
  \newcommand\continuanceref{\ref{#1}}%
  \continuancex
}{\endcontinuancex}

\hypersetup{
  pdftitle={Non-Parametric Model Calibration with  Stochastic Control Parameters},
  pdfauthor={Akshay Prasadan, Samopriya Basu, Faezeh Yazdi, Derek Bingham, Donald Estep},
  pdfkeywords={Model calibration, Inverse problems, Bayesian inference, Uncertainty quantification, Disintegration of Measure},
  colorlinks=true,
  linkcolor={blue},
  citecolor={blue},
  urlcolor={blue},
  pdfcreator={LaTeX via pandoc}}

\DeclareMathOperator{\KL}{KL}
\DeclareMathOperator{\Unif}{Unif}
\DeclareMathOperator{\Beta}{Beta}
\DeclareMathOperator*{\argmin}{arg\,min}

\newcommand{\TGDjoint}{P_{\Lambda \times U}^{\mathrm{t}}}
\newcommand{\TGDmarg}{P_{\Lambda}^{\mathrm{t}}}
\newcommand{\TGDU}{P_U^{\mathrm{t}}}
\newcommand{\SCPjoint}{\hat P_{\Lambda\times U}}
\newcommand{\SCPmarg}{\hat P_{\Lambda}}

\newcommand{\cB}{\mathcal{B}}
\newcommand{\cD}{\mathcal{D}}
\newcommand{\cN}{\mathcal{N}}
\newcommand{\cP}{\mathcal{P}}
\newcommand{\RR}{\mathbb{R}}

\begin{document}

\volume{
Volume x, Issue x, \myyear\today
}

\title{Non-Parametric Model Calibration with Stochastic Control Parameters}
\titlehead{Calibration with Stochastic Control Parameters}

\authorhead{A. Prasadan, S. Basu, F. Yazdi, D. Bingham \& D. Estep}

\corrauthor[1]{Akshay Prasadan}
\author[2]{Samopriya Basu}
\author[1]{Faezeh Yazdi}
\author[1]{Derek Bingham}
\author[1]{Donald Estep}
\corremail{akshay\_prasadan@sfu.ca}
\corraddress{Department of Statistics and Actuarial Science, 
Simon Fraser University, Burnaby, BC}
\address[1]{Department of Statistics and Actuarial Science, 
Simon Fraser University, Burnaby, BC}
\address[2]{Carleton University, Ottawa, ON}


\abstract{
We present a method for calibrating a computer model using non-parametric  techniques where the inputs are stochastic but include calibration parameters whose distributions are unknown and control parameters whose distributions are specified.  Our solution gives a distributional estimate over the input space that is consistent with observed field data, while also preserving the distribution of the known marginal of the control parameters. This property is desirable since stochastic inputs often include physical processes affecting the experimental conditions, and a scientifically plausible calibration estimate should preserve well-established distributional properties of these inputs. The method builds on recently developed non-parametric computer model calibration techniques based on the disintegration of measure and Bayesian inference.}


\maketitle

\section{Introduction}\label{section:intro}

Representing the behavior of a physical system with a computer model, e.g., the numerical solution of a partial differential equation (PDE), is common in many areas of science and engineering \cite{fang2005design}. The computer model often has two types of input: (i) those that can be measured or adjusted in the physical system, and (ii) those that are unknown to the experimenter and must be calibrated using information from field (experimental) data \cite{kennedy_ohagan}. We refer to the first type of input as control parameters\footnote{Note that the term `control' does not imply that there is latitude in its specification, but simply that it is measurable in some fashion.} and the second type as calibration parameters. 

This paper develops a new statistical methodology for computer model calibration in such situations. In our setting, both types of input are  stochastic. The distinction in the input parameters is that the calibration parameters have an unknown distribution whereas the  control parameters have a known distribution. The inferential aim is to estimate the distribution of calibration parameters given field data, the computer model, and the known distribution of control parameters.

\begin{exam}[Heat Transfer in Metal Plates] \label{example:heat:equation:intro}
The following running example illustrates a  concrete experimental setting where such problems arise. Suppose a materials engineer is interested in the thermal properties of thin, square metal plates from a supplier. These include convection and diffusion parameters which characterize how heat propagates through the plate, which are important if the metal is to undergo high temperature industrial processes. To assess these properties, a collection of plates are heated at a specified position, though realistically the placement of the flame will vary with each measurement.  The engineer records the temperature in the center of each plate after a specified period of time. This yields field data from which the thermal properties must be inferred.  

An analyst presented with the field data makes the following observations about the system.  First, the field data can be modeled as the solution to a PDE that depend on the unknown thermal properties. Learning about the thermal properties corresponds to calibrating the coefficients of this PDE. The relationship between the coefficients and the measured temperature is complex: it is not one-to-one and numerical computations can be slow. Second, the temperature measurements arise from distinct metal plates, so there is random variation in the parameters. Third, and a particular interest of this paper, certain experimental conditions do vary but are known by the engineer, perhaps only up to their distribution. For example, the position of the flame may have known Gaussian error due to imperfect physical set-ups.  These inputs parameters can be thought of as stochastic but with a known distribution, in contrast to the thermal properties. 
\end{exam}

\begin{exam} Consider a model that includes observational noise,  say $Q(\lambda,\gamma)=f(\lambda)+g(\gamma)$, for some specified $f$ and $g$ where $\lambda$ follows an unknown distribution and $\gamma$ is Gaussian with known mean and variance. In this setting, $\lambda$ is a stochastic calibration parameter and $\gamma$ a stochastic control parameter. 
\end{exam}

\begin{exam}
    A stochastic control parameter can result from known scientific processes as well. For example, the surface temperature of a body of water may vary  seasonally in a known manner and impact  hydrological data that depend on temperature.
\end{exam}

In the simpler setting where all of the stochastic inputs are unknown, i.e., only calibration parameters are present, the following inferential goal can be stated: estimate an input distribution for the calibration parameters, whose induced distribution when propagated through a computer model is consistent with field data. We refer to this as a Statistical Calibration Problem (SCP). A nonparametric framework has been developed for  the SCP using  a  Bayesian approach based on disintegration of measures \cite{breidt2011, butler2012, butler2014, butler2018combining,chi2021,bingham2024, esip_2025}. In this approach, an explicit integral formula for the solution is constructed by disintegrating  the prior and posterior probability measures with respect to the equivalence classes associated with the inverse of a computer model mapping. The procedure results in an explicit formula for the probability measure on the input space whose pushforward measure induced by the map matches the observed data distribution. 

We return to the scenario that  features stochastic control parameters which are assumed to have a specific distribution. Now the problem is to compute a solution of the SCP whose marginal on the control parameters matches the specified distribution. A straightforward application of the technique of \cite{esip_2025}, which treats the distribution of all parameters as unknown, yields a data-consistent input probability measure  on the joint parameter space, but the marginal distribution for the control variables does not, in general, match the known form. In this paper, we introduce a methodology that solves such ``constrained'' calibration problems by producing solutions that respect the known marginal distribution. This greatly extends the range of the calibration problems that can be solved using the approach in \cite{esip_2025}.

The proposed method is based on two core techniques: augmenting the output space with a component for the control parameters with known distribution, then using optimal transport to compute a joint measure on the augmented output space (from an infinite number of possibilities) that preserves the marginals on the output and the control parameters. We  apply the technique of \cite{esip_2025} on the augmented output space with the computed joint measure to estimate an input distribution solving the SCP whose marginal distribution for the control parameters matches its specified distribution. 

In \Cref{section:unconstrained:SCP}, we  review the methodology of the SCP as in \cite{esip_2025}. Then, in \Cref{section:constrained:SCP}, we introduce the augmentation step and its application in a special case where the control parameters and the corresponding field observations are paired. We then consider the more challenging case where the joint distribution of the control parameters and their field observations is unknown, which motivates the use of optimal transport. In \Cref{section:examples} we consider a simple quadratic computer model and then  simulate the heat equation setting of \Cref{example:heat:equation:intro}. We give closing remarks in \Cref{section:discussion}. In the appendix, we discuss regularity conditions and alternative formulations of our optimal transport approach. Code needed to reproduce all computational experiments and their associated data and plots is available at \url{https://github.com/akprasadan/escp-constrained}.

\section{The Statistical Calibration Problem}
\label{section:unconstrained:SCP}

A mathematical framework for the SCP has been developed in several works \cite{breidt2011, butler2012, butler2014, butler2018combining, yang2018}, and we review the most recent version \cite{esip_2025}. For any subset $A$ of Euclidean space, let $\cB_A$ denote the Borel $\sigma$-algebra restricted to $A$. We assume there is a measurable space $(\Lambda,\cB_{\Lambda})$  with $\Lambda\subset\RR^n$ compact, which represents a space of calibration parameters. Let $Q\colon \Lambda\to \cD  \coloneq Q(\Lambda)\subseteq\RR^m$  be a map representing  the behavior of a computer model. Further assume there is an unobserved probability measure $\TGDmarg$ on the input space $(\Lambda,\cB_{\Lambda})$, which we call a \textit{trial-generating distribution (TGD)}. The TGD induces a pushforward measure  $P_{\cD}=Q\TGDmarg$ on $(\cD,\cB_{\cD})$. We assume that $P_{\cD}$ is empirically observable from a random sample of measurements $\{q_i\}_{i=1}^K$. The \textit{statistical calibration problem (SCP)} is to determine a probability measure $\SCPmarg$ on $(\Lambda,\cB_{\Lambda})$ such that $Q\SCPmarg=P_{\cD}$. In other words, sampling from $\SCPmarg$ and applying the computer model should yield observations consistent with the empirical sample from $P_{\cD}$.

\begin{continuance}{example:heat:equation:intro}
Recall the experiment involves heating many metal plate samples and recording the steady state temperature at some position. The intrinsic material properties will vary across samples, impacting the distribution of recorded temperatures. The set $\Lambda$ of calibration parameters is the set of possible convection and diffusion parameters and $\TGDmarg$ is their unobserved distribution. The observed temperature is modeled as the output $Q$, which is the value of a numerical solver for the heat equation at a fixed position  assuming a perfectly placed heat source at some position. A collection of field temperature measurements $\{q_i\}_{i=1}^K\subset \cD$ will follow some  distribution $P_{\cD}$. An SCP solution $\SCPmarg$ will determine a distribution over the convection and diffusion parameters such that if we draw a sample $\{\lambda_i'\}_{i=1}^{K'}$ of parameters from $\SCPmarg$ and compute $q_i'=Q(\lambda_i')$, then the computed $\{q_i'\}_{i=1}^{K'}$ will have the same distribution as the field data $\{q_i\}_{i=1}^K$. Note that the estimated $\SCPmarg$ will almost certainly not match the latent $\TGDmarg$, because $Q$ is not one-to-one and there are many distributions consistent with $P_{\cD}$. 
\end{continuance} 

An explicit integral expression for the SCP solution is determined using disintegration of measures \cite{chang1997,esip_2025}. The disintegration of $\TGDmarg$ according to $Q$ is defined by \[\TGDmarg(A)= \int_{\cD} \int_{Q^{-1}(q)\cap A} d\TGDmarg(\lambda|q) dP_{\cD}(q)\; A \in \cB_{\Lambda},\] where $\{\TGDmarg(\cdot|q)\}_{q\in \cD}$ is a family of regular conditional probability measures on each pre-image $Q^{-1}(q)$, unique $P_{\cD}$ almost surely. But $\TGDmarg$ and therefore $\{\TGDmarg(\cdot|q)\}$ is unknown.  Fortunately, if we replace $\{\TGDmarg(\cdot|q)\}_{q\in\cD}$ with another family of probability measures on each $Q^{-1}(q)$, then the right-hand side yields another measure with the same pushforward $P_{\cD}$. To simplify this choice, we begin with a prior $P_p$, which disintegrates into a family of conditional measures $\{P_p(\cdot|q)\}$. This yields a prior-informed solution to the SCP defined by \[\SCPmarg(A)= \int_{\cD} \int_{Q^{-1}(q)\cap A} dP_{p}(\lambda|q) dP_{\cD}(q), \;A\in\cB_{\Lambda}.\] In the absence of prior knowledge, the uniform prior can be used on $\Lambda$ to obtain the posterior with maximum entropy. Otherwise, the prior should be chosen to reflect domain knowledge for a more informative solution.

With the addition of several regularity conditions on $\Lambda$, $P_p$, $Q$, and the TGD, \cite{esip_2025} demonstrate that this disintegration has several desirable properties. For example, it admits an a.e. continuous density in Lebesgue measure, with a representation in terms of parametrized surface integrals over the manifolds $Q^{-1}(q)$. Further notions of continuity with respect to the TGD are established in \cite{prasadan2026continuity}. Importantly, performing the actual disintegration of $P_p$ is unnecessary in practice, and $\SCPmarg$ can be numerically estimated with an algorithm similar to importance sampling, which is detailed in \cite[Section 4.2]{esip_2025} and also replicated in \Cref{algorithm:unconstrained} for convenience.  The necessary regularity conditions are restated in \ref{section:regularity:conditions}. We note that \Cref{algorithm:unconstrained} uses a histogram estimator for $P_{\cD}$, i.e., partitions $\cD$ into hyper-rectangles $I_i$ and estimates $P_{\cD}(I_i)$ with $\frac{|\{q_k:q_k\in I_i\}|}{K}$. Since $\cD$ is typically low-dimensional, especially compared to $\Lambda$, this is reasonable.  However, \cite[Theorem 4.6]{esip_2025} show that any estimator satisfying certain convergence assumptions, e.g., kernel density estimators, can be used. 

\begin{algorithm}[t]
\caption{Unconstrained SCP Algorithm of \cite{esip_2025}}
\begin{enumerate}
\item[\textbf{Setup:}] Fix $\Lambda\subset \RR^n$, a map $Q\colon\Lambda\to\cD$ where $\cD=Q(\Lambda)\subset \RR^m$, unknown TGD $\TGDmarg$ on $\Lambda$, 
\item[\textbf{Input}:]  Observations $\{q_i\}_{i=1}^K\sim P_{\cD}=Q\TGDmarg$, a prior density $\rho_p$ on $\Lambda$, a set $A\in\cB_{\Lambda}$ of interest
\item Partition $\cD$ into $M^m$ hyper-rectangles $I_1,\cdots, I_{M^m}$. 
\item Draw $\{\lambda_j^{(p)}\}_{j=1}^J$ from $\rho_p$
\item Compute \[\SCPmarg(A)=\sum_{i=1}^{M^{m}} \frac{|\{\lambda_j^{(p)}:\lambda_j^{(p)}\in A \text{ and }Q(\lambda_j^{(p)})\in I_i\}|}{|\{\lambda_j^{(p)}: Q(\lambda_j^{(p)})\in I_i\}|}\cdot \frac{|\{q_k:q_k\in I_i\}|}{K}.\]
\end{enumerate}
\label{algorithm:unconstrained}
\end{algorithm}

The approach to computer model calibration in \cite{esip_2025} is the natural nonparametric variation of the classic parametric approach \cite{kennedy_ohagan,higdon_2004, higdon_2008, Wang01112009, tuo2015,  plumlee17, GrosskopfBingham20,SungTuo24}. There are connections to early work in geophysics \cite{Mosegard2002,Mosegaard2006} and inverse ensemble forecasting \cite{marder2011multiple,britton2013experimentally}. Moreover, the use of disintegration is motivated by the fact that in the parametric setting, Bayes Theorem is itself an consequence of disintegration \cite{Schervish}. The nonparametric approach in \cite{esip_2025} shares a theoretical foundation with nonparametric Bayesian inference \cite{Hjort_Holmes_2010, ghosalvaart}, although a distinction is that in the latter, a prior is itself a collection of distributions, whereas \cite{esip_2025} specifies a single domain-informed prior distribution.
 
Generally, mathematical approaches to calibration employ a stochastic optimization framework involving the introduction of stochastic noise and regularization in some form \cite{Li2024inverse, Vadeboncoeura2025}. This includes the well-studied Bayesian Inverse Problem approach \cite{calvetti_2014, Stuart_Bayesian, tarantola}. The significant differences between the latter and the approach in \cite{esip_2025} is explored in \cite{bingham2024}.

\section{Input Parameters With Specified Marginal Distributions}
\label{section:constrained:SCP}

We introduce the main problem of interest, where the set of input parameters consists of a set $\Lambda\subset\RR^n$ of calibration parameters and a set $U\subset\RR^k$ of control parameters.  The TGD is now a measure $\TGDjoint$ over the product space $\Lambda \times U$, where the marginal $\TGDU(\cdot)$, defined by $\TGDU(B)=\TGDjoint(\Lambda\times B)$ for any $B\in\cB_U$, is known or specified. However, the joint TGD $\TGDjoint$ and its marginal $\TGDmarg$ in $\Lambda$ are unknown. The computer model is of the form $Q\colon\Lambda\times U\to\cD$, where $\cD\subset \RR^m$ and $m\le n+k$. We again denote the pushforward by $Q$ of the TGD by $P_{\cD}\coloneq Q\TGDjoint$. As before, $P_{\cD}$ is observed with field data $\{q_i\}_{i=1}^K$.

\begin{continuance}{example:heat:equation:intro}
    Previously, we assumed the position of the heat source is known and deterministic with no human or machine error, and the only stochastic inputs are the thermal properties. Now, $U$ represents the space of possible positions, which for this illustration is assumed to have some known distribution $\TGDU$, e.g., Gaussian.  
\end{continuance}

The SCP solution in this new setting should be a probability measure on $\Lambda\times U$ that is both consistent with the observed data and with the known control parameter distribution $\TGDU$ to be scientifically plausible. This motivates the following abstract modification of the SCP:

  \begin{defi}  The \textit{statistical calibration problem with partially known stochastic input distributions}: determine a probability measure $\SCPjoint$ on $\Lambda\times U$ such that $Q\SCPjoint=Q \TGDjoint\eqcolon  P_{\cD}$ and $\pi_U \SCPjoint=\pi_U \TGDjoint\eqcolon \TGDU$, where $\pi_U\colon \Lambda\times U\to U$ is the orthogonal projection onto $U$. As a shorthand, we  refer to this as a \textit{constrained SCP}. If we drop the $\pi_U\SCPjoint=\TGDU$ requirement, we refer to the resulting problem as an \textit{unconstrained SCP} on $\Lambda\times U$. This is equivalent to the SCP described in \Cref{section:unconstrained:SCP}.
  \end{defi}

 Na\"ively treating the new problem as an unconstrained SCP,  we could specify  a prior distribution $P_p$ on $\Lambda\times U$, disintegrate it using $Q$, and obtain the solution \[\SCPjoint(E)= \int_{\cD} \int_{Q^{-1}(q)\cap E} dP_{p}(\lambda, u|q) dP_{\cD}(q),\; E\in\cB_{\Lambda}\otimes \cB_{U},\] where $Q^{-1}(q)$ is  a contour in $\Lambda\times U$ and $P_p(\lambda,u|q)$ is a conditional measure over $Q^{-1}(q)$. While $\SCPjoint$ is assured to have the correct pushforward $P_{\cD}$ by the results of \cite{esip_2025}, no such guarantee is possible for the $U$ marginal, which requires that for all $B\in\cB_U$, \[\TGDU(B)=\SCPjoint(\Lambda\times B) =\int_{\cD} \int_{Q^{-1}(q)\cap (\Lambda\times B)} dP_{p}(\lambda, u|q) dP_{\cD}(q).\] 

\begin{exam}[Limitations of the Unconstrained SCP]
\label{example:counterexample}

 We  illustrate a simple example where the solution of the unconstrained SCP lacks the desired $U$ property. Let $\Lambda = [-10, 10]\times[-10,10]$, $U=[-20, 20]$, $Q(\lambda_1,\lambda_2,u)=\lambda_1^2+\lambda_2^2+u^2$. Assume the unobserved TGD $\TGDjoint$ is the mixture $\frac{1}{2}\cN(\mu_A,\Sigma_A) +\frac{1}{2}\cN(\mu_B,\Sigma_B)$ where \[\mu_A=\begin{bmatrix}
     3, & 3,  &6
\end{bmatrix}^T,\, \mu_B=\begin{bmatrix}
     -3, & -3,  &-3
\end{bmatrix}^T, \, \text{and }\Sigma_A =\Sigma_B= \mathrm{diag}(1, 1, 9).\] Note that $\TGDU=\frac{1}{2}\cN(6, 9)+\frac{1}{2}\cN(-3, 9)$. We assume a prior distribution $P_p=\Unif[-10, 10]\times\Unif[-10,10] \times \TGDU$, reflecting perfect knowledge of the $U$ marginal but assuming little on $\Lambda$ other than its domain. We simulate $\{q_i\}_{i=1}^{K}$ for $K=\num{20,000}$ by drawing $(\lambda_i, u_i)\sim \TGDjoint$ and setting $q_i=Q(\lambda_i,u_i)$.  Then we use \Cref{algorithm:unconstrained} on the $\{q_i\}$ with the prior distribution $P_p$ (having discarded the $\{(\lambda_i, u_i)\}_{i=1}^K$ samples). This yields a solution $\SCPjoint$, from which we marginalize out $\Lambda$ and obtain a measure $\hat P_U$. 

As shown in the top-left plot of \Cref{fig:counterexample}, the $U$ marginal of the unconstrained solution overweights the left mode as well as the domain between the modes. On the other hand, using the constrained method we  discuss  below (\Cref{algorithm:constrained:unpaired}) recovers the expected symmetric bimodality. The top-right plot shows that the pushforward by $Q$ for both SCP solutions matches the observed pushforward, as guaranteed by both methods. The bottom two plots show the SCP solutions for the unconstrained and constrained, respectively, marginalized to $\Lambda$. Specifically, for any SCP solution $\SCPjoint$, we create a grid of rectangles $\{A_{ij}\}$ partitioning $\Lambda$, compute $\SCPjoint(A_{ij}\times U)$, and map the probability values to a color gradient. The constrained method (bottom right) yields high densities on contours of $Q$ that pass through the unobserved values in $\Lambda$, whereas the unconstrained method's density (bottom left) is more diffuse. Such contours are to be expected, since they represent the regions of TGD values, weighted by the prior distribution, that are consistent with the observed values in $\cD$. We  discuss the construction and interpretation of these SCP solutions in greater detail in  \Cref{section:examples}. 

\begin{figure}[t]
    \centering
    \includegraphics[width=0.7\textwidth]{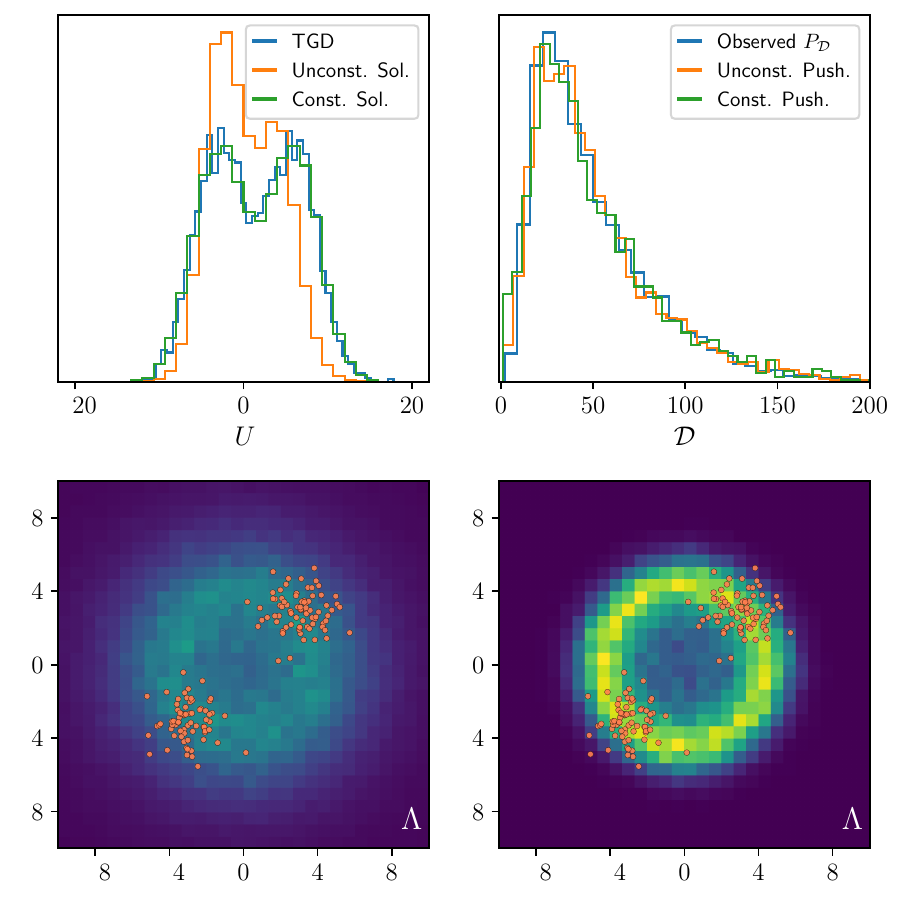} 
    \caption{We present the results of the simulation of \Cref{example:counterexample}. Top left: The $U$ marginal of three distributions: that of the TGD $\TGDjoint$, the unconstrained SCP (\Cref{algorithm:unconstrained}), and the proposed constrained SCP using unpaired data (\Cref{algorithm:constrained:unpaired}). Top right: the pushforwards by $Q$ of the TGD, unconstrained SCP, and proposed constrained SCP. Bottom left: A heatmap depicting a discrete estimate of the unconstrained SCP on $\Lambda$ after marginalizing out $U$. Bottom right: Same as bottom left, but using the proposed constrained SCP. The same color scale is used for both of the bottom plots. The construction of the heatmaps is explained in greater detail in \Cref{section:examples}.}
    \label{fig:counterexample}
\end{figure}
\end{exam}

\subsection{Augmenting the Computer Model}

Observe that the characterization of the SCP with a known marginal depends on matching two different pushforwards, one induced by $Q$, and the other by $\pi_U$. This motivates introducing an \textit{augmented} map $\tilde Q$ for which preserving its pushforward preserves both aforementioned pushforwards automatically. Hence, we define $\tilde Q\colon\Lambda\times U \to \cD\times U$ by \begin{equation} \label{eq:augmented:Q}
\tilde{Q}(\lambda,  u)\coloneq \begin{bmatrix}
    Q(\lambda, u)\\ \pi_U(\lambda, u)
\end{bmatrix} = \begin{bmatrix}
    Q(\lambda, u)\\ u
\end{bmatrix}.\end{equation} By construction, the pushforward $\tilde Q \TGDjoint$ has marginals $P_{\cD}\coloneq Q\TGDjoint$ and $\TGDU\coloneq \pi_U \TGDjoint$ on $\cD$ and $U$, respectively. The following theorem demonstrates that this augmentation provides a way to solve the partially-specified SCP.

\begin{theorem} Suppose $\SCPjoint$ is a probability measure on $\cB_{\Lambda}\otimes\cB_{U}$  such that $\tilde Q \SCPjoint =\tilde Q  \TGDjoint$. Then $\SCPjoint$ solves the SCP with a partially known stochastic input distribution, i.e., $Q\SCPjoint=Q\TGDjoint$ and $\pi_U\SCPjoint=\pi_U \TGDjoint$.
\end{theorem} 
    \begin{proof}
        Pick $C\in\cB_\cD$. Observe $Q^{-1}(C)=\tilde Q^{-1}(C\times U)$, so \[(Q\SCPjoint)(C) =(\tilde Q \SCPjoint)(C\times U)= (\tilde Q \TGDjoint)(C\times U) = (Q\TGDjoint)(C).\] Next, pick $B\in \cB_{U}$. Then using $\pi_U^{-1}(B)=\Lambda\times B = \tilde Q^{-1}(\cD\times U)$, \begin{align*}
    (\pi_U\SCPjoint)(B) &=\SCPjoint(\Lambda\times B)= (\tilde Q \SCPjoint)(\cD\times U) = (\tilde Q  \TGDjoint)(\cD\times U) \\ &= \TGDjoint(\Lambda\times B)=\TGDU(B).
\end{align*}
    \end{proof}

Thus, finding  a measure $\SCPjoint$ that matches the pushforward of the TGD by $\tilde Q$ solves the SCP with a known marginal. Consequently, we have reduced the constrained SCP to an unconstrained SCP with the augmented computer model, meaning we can directly apply the method of \cite{esip_2025}. The disintegration  involves contours of $\tilde Q^{-1}(q, u)$ for $(q,u)\in\cD\times U$. That is, we numerically estimate via \Cref{algorithm:unconstrained} \begin{equation} \label{eq:disintegration:constrained}
    \SCPjoint(E)= \int_{\cD\times U} \int_{\tilde Q^{-1}(q, u)\cap E} dP_{p}(\lambda', u'|q, u) \,d (\tilde Q \TGDjoint)(q,u),\; E\in\cB_{\Lambda}\otimes \cB_{U}. 
\end{equation}The prior distribution $P_p$ is  placed on $\Lambda\times U$, and since the $U$ marginal is known, we suggest setting $P_p$ to have the density $\rho_p=\rho_{\Lambda} \rho_U^{\mathrm{t}}$, where $\rho_{\Lambda}$ is a prior density on $\Lambda$ and $\rho_U^{\mathrm{t}}$ is the  density of $\TGDU.$ Such densities exist from the regularity conditions in the appendix.

However, one complication remains: the unconstrained SCP requires knowledge of or samples from the pushforward $\tilde Q \TGDjoint$. To be clear, we do always have access to the pushforwards $P_{\cD}=Q\TGDmarg$ and $\TGDU=\pi_U \TGDjoint$, but these are only the marginal distributions of $\tilde Q \TGDjoint$. Practically, the distinction corresponds to whether the values of control parameters are recorded alongside the field observations. We thus consider two cases.

\subsection{Paired Data in the Augmented Space}
\label{subsection:paired:data}
 In the simpler case, we record values as pairs $\{(q_i, u_i)\}_{i=1}^K$ where each $q_i= Q(\lambda_i, u_i)$ for some $(\lambda_i, u_i)\sim \TGDjoint$. The pairs $\{(q_i, u_i)\}_{i=1}^k$  are thus distributed according to $\tilde Q \TGDjoint$. The $\{\lambda_i\}_{i=1}^K$ are unobserved.  In \Cref{example:heat:equation:intro}, this case arises if the  location of the heat source is recorded alongside the temperature in each trial. If such paired data or distributional knowledge is known, then the methods of \cite{esip_2025} immediately apply, assuming the regularity conditions stated in \ref{section:regularity:conditions}. The explicit procedure is given in \Cref{algorithm:constrained:paired}, which is equivalent to \Cref{algorithm:unconstrained} with the suitable modifications of the domain, computer model, TGD, and prior distribution. As explained  for \Cref{algorithm:unconstrained}, we may replace the histogram estimator of $P_{\cD}$ with a kernel density estimator, for example.

\begin{algorithm}[t]
\caption{Constrained SCP Algorithm with Paired Data}
\begin{enumerate}
\item[\textbf{Setup:}] Fix $\Lambda\times U\subset \RR^n\times \RR^k$, a map $Q$ and its augmentation $\tilde Q$, $\cD=Q(\Lambda\times U)\subset\RR^m$,  unknown TGD $\TGDjoint$ on $\Lambda\times U$,  known marginal TGD $\TGDU$ 
\item[\textbf{Input}:]  Observations $\{(q_i, u_i)\}_{i=1}^K\sim \tilde Q \TGDjoint$, prior density $\rho_p=\rho_{\Lambda}\rho_U^\mathrm{t}$ on $\Lambda\times U$,  set $A\in\cB_{\Lambda\times U}$ of interest
\item Partition $\cD\times U$ into $M^{m+k}$ hyper-rectangles $I_1,\cdots, I_{M^{m+k}}$.
\item Draw $\{(\lambda_j^{(p)}, u_j^{(p)})\}_{j=1}^J$ from $\rho_p$. Set $\eta_j^{(p)}= (\lambda_j^{(p)}, u_j^{(p)})$ for shorthand.
\item Compute \begin{align*}
    \SCPjoint(A) = \sum_{i=1}^{M^{m+k}} \frac{|\{\eta_j^{(p)}:\eta_j^{(p)}\in A \text{ and }\tilde Q(\eta_j^{(p)})\in I_i\}|}{|\{\eta_j^{(p)}: \tilde Q (\eta_j^{(p)})\in I_i\}|} \cdot \frac{|\{(q_k, u_k):(q_k,u_k)\in I_i\}|}{K}.
\end{align*}
\end{enumerate}
\label{algorithm:constrained:paired}
\end{algorithm}

\subsection{Unpaired Data in the Augmented Space}
\label{subsection:unpaired:data}
In the more challenging scenario, such paired data are not available. Instead, we are only given data $\{q_i\}_{i=1}^K$ where $q_i=Q(\lambda_i, u_i)$ for latent pairs $\{(\lambda_i, u_i)\}_{i=1}^K$. However, we still retain knowledge of $\TGDU$ and can generate samples from it.  We  rely on techniques from optimal transport to generate paired data from a surrogate of the joint pushforward $\tilde Q \TGDjoint$. For simplification of the presentation, we assume $U\subset \RR$ and $D\subset\RR$, i.e., $m=k=1$. This is not a requirement; see \cite{multimarginal:optimal:transport} for a discussion of multi-marginal optimal transport techniques. 

We must specify a joint distribution on $\cD\times U$ with known marginals $P_{\cD}$ and $\TGDU$. There are several techniques that accomplish this goal, including copulas \cite{nelson2006} and optimal transport \cite{santambrogio2015, villani2008}. We  focus on optimal transport, but our method does not require this particular choice. In optimal transport, we choose a joint distribution with given marginals such that its average `cost', defined by some loss function, is minimized. For example, if a cost function on $\cD\times U$ is $c(q,u)=|q-u|^2$,  the chosen joint measure is concentrated in regions along the diagonal $q=u$, to the degree permitted by the constrained marginals. Since the joint pushforward measure  is induced by $\tilde Q$, we choose a cost function incorporating $Q$ directly.

Let $\mathfrak{F}\coloneq \mathfrak{F}(P_{\cD},\TGDU)$ denote the set of joint probability measures on $\cD\times U$ with the desired marginals. Clearly $\mathfrak{F}$ is non-empty since it includes the independent coupling $P_{\cD}\times \TGDU$, among others. For a cost function $c\colon \cD\times U\to\RR$ and penalty hyper-parameter $\varepsilon$, the \textit{entropic regularized Kantorovich problem} \cite{nutz2021} solves \begin{equation} \label{eq:kantorovich:regularized}
    \hat \Pi = \argmin_{\Gamma\in\mathfrak{F}}\left[\int_{\cD\times U}c(q,u) d\Gamma(q, u) + \varepsilon\cdot \KL(\Gamma, P_{\cD}\times \TGDU)\right].
\end{equation} We denote by $\KL$ the Kullback-Leibler divergence operator, which for any two probability measures $\nu_1,\nu_2$ on a space $\mathcal{X}$ is defined as
\[\KL(\nu_1,\nu_2)\coloneq \begin{cases}
    \int_{\mathcal{X}}\log\frac{d \nu_1}{d  \nu_2}d  \nu_1, &\text{if }\nu_1\ll\nu_2\\
    \infty, &\text{if }\nu_1\not\ll\nu_2.
\end{cases}\]  The $\KL$ operator specifies a notion of distance between two probability measures. Thus,
the first term of \Cref{eq:kantorovich:regularized} penalizes $\Gamma\in\mathfrak{F}$ according to the cost function, while the second term penalizes deviations of $\Gamma$ from the independent coupling.  To incorporate the geometry of $Q$, we propose setting \begin{equation} \label{eq:cost:function}
    c(q,u)=\int_{\Lambda} |q-Q(\lambda, u)|^2 \rho_{\Lambda}(\lambda) d\lambda,
\end{equation} where $\rho_{\Lambda}$ is the prior density on $\Lambda$ taken from our prior density $\rho_p=\rho_{\Lambda}\rho_U^\mathrm{t}$ on $\Lambda\times U$. Using the cost function \Cref{eq:cost:function} means $\hat\Pi$ place more weight near the output values $q=Q(\lambda,u)$, where $\lambda$ is drawn from some prior and $u\sim \TGDU$. Without the $\KL$ term,  $\hat\Pi$ would overfit to the graph of $Q$ and not admit a density in Lebesgue measure. Regularization controls this overfitting, providing computational benefits, and ensures that $\hat\Pi$ has a density.

The following result states that the choice of cost function \Cref{eq:cost:function} is sufficient  for the existence of a solution to the regularized Kantorovich problem in \cref{eq:kantorovich:regularized}. The proof follows from \cite[Theorem 4.2]{nutz2021}, noting the function $c$ in  \cref{eq:cost:function} is continuous over a compact set.

\begin{theorem} \label{theorem:EOT:solution}
   For any $\varepsilon>0$, \cref{eq:kantorovich:regularized} admits a solution.
\end{theorem}

Estimation of $c(q,u)$ is performed via standard Monte Carlo techniques. Once $c$ is computed, we solve for $\hat \Pi$ using the Sinkhorn-Knopp (SK) algorithm (\cite{sinkhorn_1967,nutz2021}).  We use the POT: Python Optimal Transport package for the SK computations \cite{flamary2021pot, flamary2024pot}. Once  a $\hat \Pi$ is constructed, we generate synthetic field data by drawing samples $(q_i, u_i)\sim \hat \Pi$, which are assured to have marginals $P_{\cD}$ and $\TGDU$.  Then we solve the unconstrained SCP with the augmented $\tilde Q$ as described for the paired data. The full details are given in \Cref{algorithm:constrained:unpaired}. 

\begin{algorithm}[t]
\caption{Constrained SCP Algorithm with Unpaired Data}
\begin{enumerate}
\item[\textbf{Setup:}] Same as \Cref{algorithm:constrained:paired}, but $\cD\subset\RR$ and $U\subset\RR$ 
\item[\textbf{Input}:]  Observations $\{q_i\}_{i=1}^{K_1}\sim P_{\cD}$ and $\{u_i\}_{i=1}^{K_2}\sim \TGDU$, prior density $\rho_p=\rho_{\Lambda}\rho_U^\mathrm{t}$ on $\Lambda\times U$, SK regularization hyperparameter $\varepsilon$, and a set $A\in\cB_{\Lambda\times U}$ of interest
\item Construct an estimate $\hat\Pi$ of $\tilde Q \TGDjoint$.
    \begin{enumerate}
       \item Fix bin counts $M_1$ and $M_2$ on $\cD$ and $U$ respectively. 
    \item Estimate a normalized histogram density on $\cD$ using $M_1$ bins, and on $U$ using $M_2$  bins, using the respective data $\{q_i\}_{i=1}^{K_1}$ and $\{u_i\}_{i=1}^{K_2}$.  \label{step:alg:histogram}
    \item Form grid of $M_1\times M_2$ points $\{(q_r^{(c)}, u_s^{(c)})\}_{r,s}$ over $\cD\times U$. For each $(q_r^{(c)}, u_s^{(c)})$, draw $K_3$ points $\{\lambda_{j, r,s}^{(p)}\}$ from $\rho_\Lambda$. Estimate $c(q_r^{(c)}, u_s^{(c)})$ defined in \eqref{eq:cost:function} by the Monte Carlo estimate \[\frac{1}{K_3}\sum_{j=1}^{K_3}|q_r^{(c)} - Q(\lambda_{j,r,s}^{(p)}, u_s^{(c)})|^2.\] Form an $M_1\times M_2$ cost matrix $\mathbf{C}=\{(c(q_r^{(c)}, u_s^{(c)})\}_{r,s}$.
    \item Run SK's algorithm using  \Cref{step:alg:histogram}'s empirical distributions, $\mathbf{C}$, and the regularization hyperparameter $\varepsilon$ to get a measure $\hat\Pi$ on $\cD\times U$.
    \end{enumerate}

\item Draw $\{(q_i', u_i')\}_{i=1}^{K_4}\sim \hat\Pi$. Carry out steps (1)--(3) of \Cref{algorithm:constrained:paired} using this paired data as input. 
\end{enumerate}
\label{algorithm:constrained:unpaired}
\end{algorithm}

This procedure depends on a choice of hyperparameter $\varepsilon$ in the SK step and tuning is required.  Experimentally, we find that the SCP solution tends to be very stable for a broad range of $\varepsilon$, even as the SK solution shows broad variation in the shape of the joint distribution. However, taking $\varepsilon$ to be too small leads to numerical issues in SK and thereby causes downstream issues in the SCP solution. We also found in these examples that using unpaired data and the SK algorithm yields similar estimates as  paired data, even though we are using  a joint distribution that certainly differs from the true distribution.

The authors considered an alternative formulation of the Kantorovich problem \Cref{eq:kantorovich:regularized}. For example, we can integrate out $\Lambda$ subsequent to the minimization problem. That is, we  set $c_{\lambda}(q,u)=|q-Q(\lambda, u)|^2$, solve the corresponding regularized Kantorovich problem to obtain a $\lambda$-dependent solution $\hat \Pi_{\lambda}$, and finally set $\hat \Pi=\int_{\Lambda} \hat\Pi_{\lambda} \rho_{\Lambda}(\lambda) d\lambda$ using a prior density $\rho_{\Lambda}$ on $\Lambda$. A similar argument as \Cref{theorem:EOT:solution} shows each minimization problem has a solution, and we prove in \ref{section:measurability} that $\lambda\mapsto \hat\Pi_{\lambda}$ is measurable, which makes the averaging operation well-posed. A simple check shows that $\hat \Pi$ has marginals $P_{\cD}$ and $\TGDU$ as desired. The disadvantage of this approach is having to repeatedly re-run the SK algorithm for a large number of $\lambda\in\Lambda$. 

\section{Examples}
\label{section:examples}

We present two experiments. In each, we solve both the paired and unpaired constrained SCP, using the paired case  as a benchmark for comparison. We  consider $\Lambda$ of dimensions 2 and 3 and both $U$ and $\cD$ of dimension 1. In the first experiment,  $Q$ is a  quadratic function with an additive noise term representing a stochastic control parameter. In the second experiment, $Q$ is the solution of a heat equation  based on \Cref{example:heat:equation:intro}, where the  control parameter is the vertical position of the heat source. Before we discuss the specific experiments, we explain the general computational procedure and visualization techniques.

\subsection{Computational Details}
\label{subsection:computational:details}

To create synthetic data for the experiments, we  draw pairs $\{(\lambda_i, u_i)\}_{i=1}^{K_1}$ from the TGD (which is unknown in applications), set $q_i=Q(\lambda_i, u_i)$, and discard the $\{\lambda_i\}_{i=1}^{K_1}$. In the paired case, we treat the $\{(q_i,u_i)\}_{i=1}^{K_1}$ as the observations. We use \Cref{algorithm:constrained:paired} to produce an estimated measure $\SCPjoint$ on $\Lambda\times U$ whose pushforward by $Q$ matches the distribution $P_{\cD}$ of the $\{q_i\}$ and whose $U$ marginal is $\TGDU$. We always choose the prior distribution of the independent coupling of the uniform distribution over $\Lambda$ and the known $\TGDU$. This algorithm yields a measure $\SCPjoint$ on $\Lambda\times U$.

To simulate the unpaired case, we compute $\{q_i\}_{i=1}^{K_1}$ as above then discard the pairs $\{(\lambda_i, u_i)\}_{i=1}^{K_1}$. Next, we draw a new independent sample $\{\tilde u_j\}_{j=1}^{K_2}$, and the two unpaired datasets $\{q_i\}_{i=1}^{K_1}$ and $\{\tilde u_j\}_{j=1}^{K_2}$ form the simulated observations. Then we run \Cref{algorithm:constrained:unpaired} with these inputs to obtain another estimate $\SCPjoint$, with the same pushforward $P_{\cD}$ as before and the correct $U$ marginal. 

It is important to distinguish the  sample size from algorithmic hyper-parameters. The integers $K$ in \Cref{algorithm:constrained:paired} and $K_1$ in \Cref{algorithm:constrained:unpaired} represent the  number of field observations. By contrast, the other quantities are at the analyst's full discretion. In each SCP scenario, the number of prior samples $J$ used in the importance sampling step is independent of the observation count, and for the unpaired case, the same is true for the numbers $K_2$ of independent draws from $\TGDU$, $K_3$ of Monte Carlo samples, and $K_4$ of pairs from the SK estimate. The main limiting factor is the computational burden of evaluating $Q$ and drawing samples from a prior. For example, in the SK step, we must evaluate the cost function at $M_1M_2$ points $\{(q_r^{(c)}, u_s^{(c)})\}_{r,s}$, each of which requires $K_3$ evaluations of $Q$ per point. If $Q$ is computationally expensive to evaluate we compute a surrogate model for the computer model itself \cite{gramacy2020surrogates}.  In the following examples, we generally choose $K$ and $K_1$ to be smaller than the other integers to be realistic about the amount of data available. We also demonstrate the use of a surrogate model in the heat equation example to enable larger choices of $K_3$, $K_4$, and $J$.

To visualize the SCP solutions, we partition the input space $\Lambda\times U$ into a relatively fine grid of hyper-rectangles, e.g., 50 bins per axis, and  compute $\SCPjoint$ on each hyper-rectangle. 
For a 2-dimensional $\Lambda$, if this tiling of $\Lambda\times U$ is  $\{A_{ijk}\}_{1\le i,j,k\le 50}$, we  apply the formula in the final step of \Cref{algorithm:constrained:unpaired} to  estimate $\{\SCPjoint(A_{ijk})\}$. Marginalizing this over $U$, we compute $\SCPmarg(A_{ij\,\star})\coloneq \sum_{k=1}^{50}\SCPjoint(A_{ijk})$, yielding a discrete probability mass function on $\Lambda$ approximating the marginal density of $\SCPjoint$. We  visualize this with a heatmap. To give comparison to the TGD, the plots include a subset of points from the TGD projected onto $\Lambda$. We note that these points are not used in the SCP algorithm.  When comparing the paired and unpaired case heatmaps, we will always use the same color scale.

In this 2-dimensional $\Lambda$ case, we also visualize the full SCP solution on $\Lambda\times U$. To do so, for each $1\le k\le 50$, we construct a heatmap visualizing the distribution of probabilities $\{P_{\Lambda\times U}(A_{ijk})\}_{1\le i,j\le 50}$. Then, we stack these heatmaps for each value of $k$ along the $U$ axis. This yields a discrete representation of the 3-dimensional density. As before, we  include a subset of points from the TGD. For 3-dimensional $\Lambda$, we plot the marginal SCP solutions for each pair of dimensions in $\Lambda$, and for the stacked heatmap visual, we marginalize out $U$ and visualize the solution over the 3-dimensional $\Lambda$ space.

For diagnostic purposes, we  compute and visualize  associated distributions on $\cD$ and $U$. For $U$, we plot the marginal distribution of $U$ from the unpaired and paired SCPs, the distribution of the $U$ marginal from the TGD, and the marginal $U$ distribution of samples from the SK solution. We compute the $U$ marginals of the SCP solutions by evaluating $\hat P_U(A_{*k})\coloneq\sum_{i,j}\SCPjoint(A_{ijk})$ for each $k$, i.e., marginalizing out the $\Lambda$ portion of the tiling.  The other marginals are computed as histogram estimates from samples drawn from the associated TGD or SK solution. All of these distributions should line up, indicating  the SK estimation was successful and the SCP solution is  constrained as desired. For $\cD$, we plot the pushforwards of both SCP solutions by $Q$, the pushforward of the TGD, and again the marginal $\cD$ distribution of samples from the SK solution. The three pushforwards are generated by resampling from the SCP solution heatmap or TGD, applying $Q$, and computing a histogram estimate. Again, each distribution should match. We also include two joint distributions of $\cD\times U$, one for the joint SK solution and one for the true pushforward $\tilde Q\TGDjoint$ (i.e., the distribution of paired data $\{(q_i,u_i)\}$), along with some actual points from these distributions.

\begin{exam}[Quadratic Model with Additive Noise]
\label{exam:euclidean:norm}

\begin{figure}[t]
    \centering
    \begin{minipage}{0.35\textwidth}
        \centering
        \includegraphics[width=\linewidth]{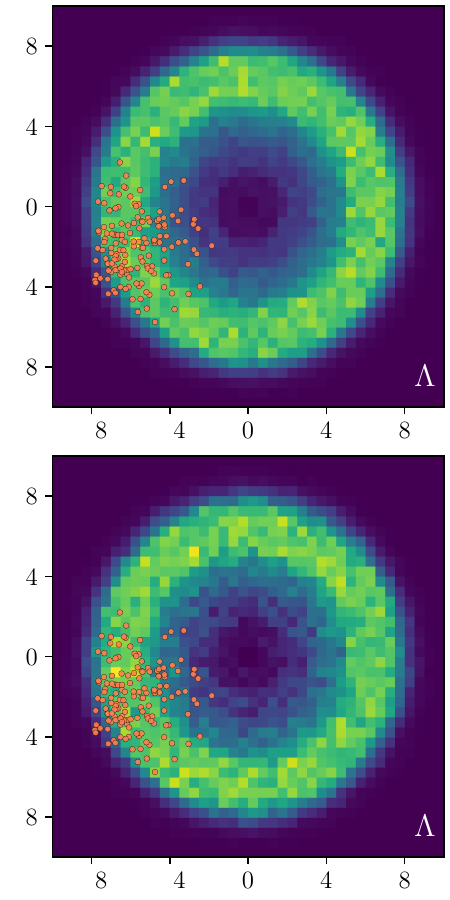}
    \end{minipage}\hfill
    \begin{minipage}{0.65\textwidth}
        \centering
        \includegraphics[width=\linewidth]{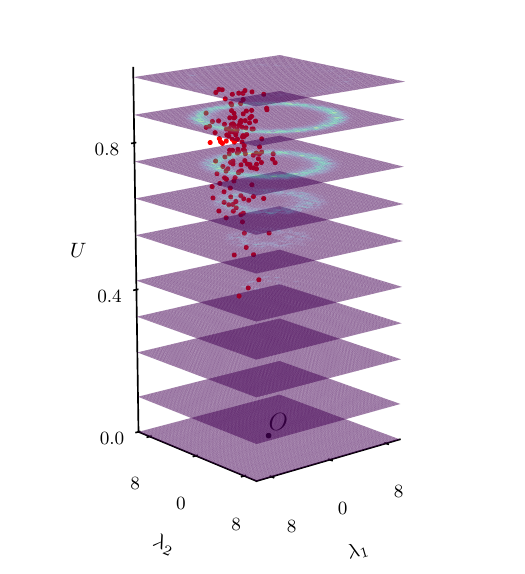}
    \end{minipage}
    \caption{Empirical estimates of constrained SCP solutions for \Cref{exam:euclidean:norm}. Lighter colors correspond to higher probabilities. The same color scale is used on the two left-plots. A scatterplot of points is included from the unobserved TGD; the same points appear in each plot.  Top left: Paired case, $\Lambda$-marginal. Bottom left: Unpaired case, $\Lambda$ marginal. Right: Joint estimate of the unpaired case's $\SCPjoint$.}
    \label{fig:euclidean:norm:sol}
\end{figure}

\begin{figure}[t]
    \centering
    \includegraphics[width=0.7\textwidth]{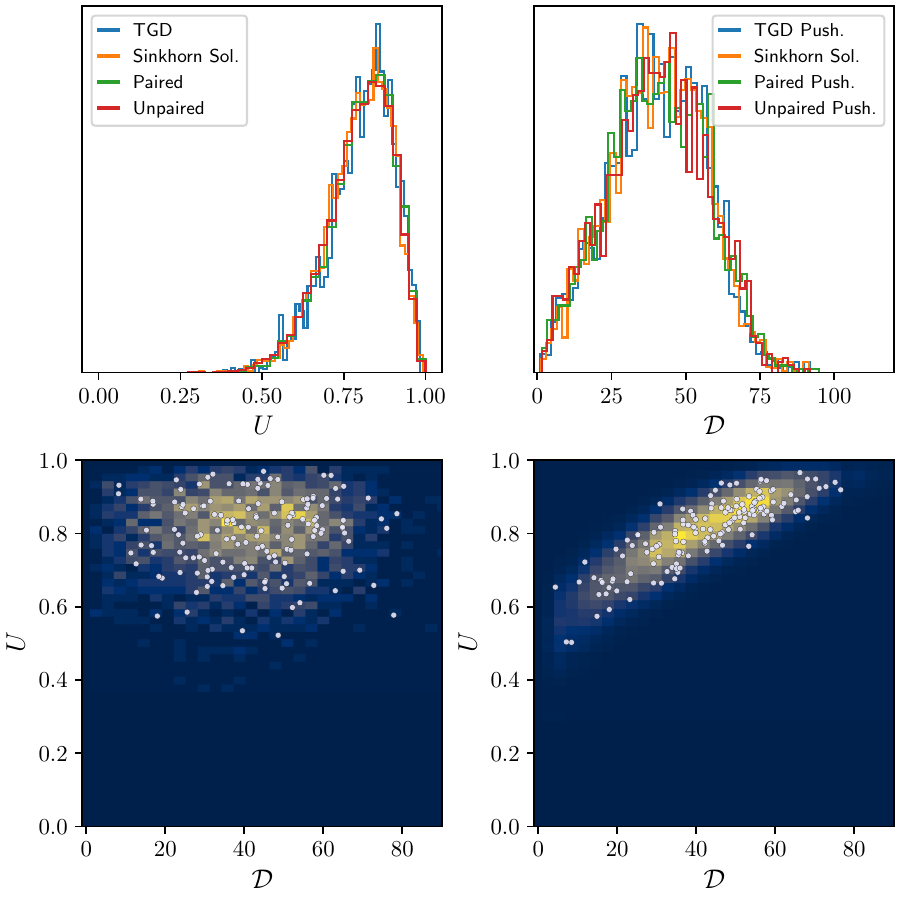}
    \caption{Associated $\cD$ and $U$ distributions from \Cref{exam:euclidean:norm}. Top left: a histogram plot of the marginal $U$ distribution of the TGD, SK solution, and both paired and unpaired constrained SCP solutions. Top right: histogram plots of the pushforward by $Q$ of the TGD and both paired and unpaired constrained SCP solution, as well as the $\cD$ marginal of the SK solution. Bottom left: a heatmap of an empirical estimate of the true joint distribution $\tilde Q \TGDjoint$ on $\cD\times U$ with some samples drawn from this distribution plotted as points. Bottom right: A heatmap of the SK estimate $\hat\Pi$ and some samples drawn from it plotted as points.}
    \label{fig:euclidean:norm:diagnostics}
\end{figure}

Set $\Lambda = [-10, 10]\times[-10,10]$ and $U=[0, 1]$. Set $Q(\lambda_1,\lambda_2,u)=\lambda_1^2+\lambda_2^2+u$. Let $\TGDjoint=\nu_1\times\nu_2\times \TGDU$, where $\nu_1\sim12\cdot\Beta(2, 8)-8$, $\nu_2\sim12\cdot\Beta(4, 4)-7$, and $\TGDU\sim\Beta(12, 3)$. In other words, $\TGDjoint$ is the independent coupling of three distinct Beta distributions scaled and translated to lie on $[-8, 4]\times[-8,4]\times[0, 1]$.  The prior distribution is chosen to be $P_p=\Unif[-10, 10]\times\Unif[-10,10]\times  \TGDU$, noting that the prior distribution does not take into account that the marginal TGD $\nu_1\times \nu_2$  concentrates on a strict subset $(-8, 4]\times(-8,4]$ of $\Lambda$.

For the paired case, in the notation of \Cref{algorithm:constrained:paired}, we use $K=\num{3,000}$ simulated observations, $M=30$ bins per axis of $\cD$ and $U$, and $J=\num{250,000}$ samples from the prior. For visualization of the SCP solution density, we partition each dimension of $\Lambda\times U$ into $40$ bins. For the unpaired case and in terms of \Cref{algorithm:constrained:unpaired}, we use $K_1=\num{3,000}$ observations, $K_2=\num{100,000}$ independent draws from $\TGDU$, $K_3=\num{15,000}$ Monte Carlo draws, $M_1=M_2=30$ bins on $\cD$ and $U$ for SK with $\varepsilon=1$, and $K_4=\num{100,000}$ samples from the SK solution $\hat\Pi$. When re-running the steps of the paired algorithm, we use $J=\num{250,000}$ prior samples, $M=30$ bins on $\cD$ and $U$, and visualize with $40$ bins per axis in $\Lambda\times U.$

We depict the solution in \Cref{fig:euclidean:norm:sol}. On the top-left, we visualize as a heatmap the marginal distribution of $\SCPjoint$ on $\Lambda$ for the paired data case, and the same for the unpaired data on the bottom left. The added points are a sample from the TGD, projected onto $\Lambda$ (the same points are used in either plot). As remarked earlier, the cost of having unpaired data is surprisingly minimal. On the right, we show a stacked heatmap representing the full distribution of $\SCPjoint$ as well as a sample of points. 

The shape of the solution follows from the geometry of $Q^{-1}(q)$ (a sphere) and the empirical estimator used in \Cref{algorithm:unconstrained} and its constrained extensions. The solution's density overweights regions of the input space that simultaneously have high prior weight and pass through the contours $\{Q^{-1}(q)\}$ (or the empirical analogue $\{Q^{-1}(I_i)\}$ from the numerical algorithm). We used a uniform prior on $\Lambda$, so all portions of the contours are given equal priority. Alternatively, if we chose a prior over $\Lambda$ that is concentrated  on the quadrant $[0,8]\times[0,8]$, then the annulus shape on the left-plots would be replaced by the subset  in that quadrant. We cannot expect the solution to solely concentrate near the latent TGD points because the algorithm only has access to samples in $\cD$, so it can only infer the fibers along which the input points might lie. Inspecting the plots on the right, we see that the bulk of the mass occurs for higher values of $U$, consistent with the $\Beta(12,3)$ marginal.

As a diagnostic check, we plot in \Cref{fig:euclidean:norm:diagnostics} associated marginal and joint distributions related to $\cD\times U$. All $U$ marginals and pushforwards in $\cD$ match as expected, and the SK solution has the correct marginals in $\cD$ and $U$ despite differing from the true joint.
\end{exam}

\begin{exam}[Thin-Plate Heat Equation]
\label{exam:thin:plate:eq}

\begin{figure}[t]
    \centering  
    \includegraphics[width=0.8\textwidth]{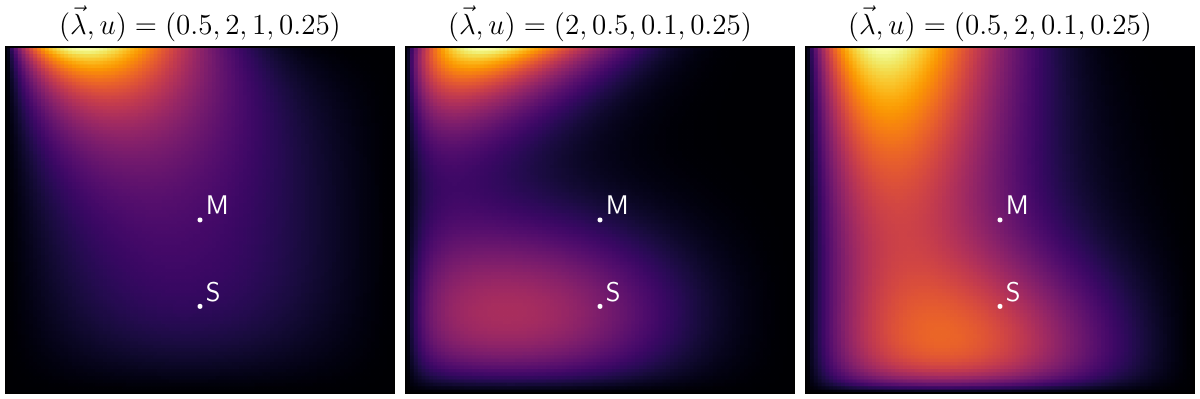}
    \caption{Steady state temperature profiles on a thin metal plate for three different combinations of convection and diffusion parameters and vertical position $u$ of the heat source, as defined in \Cref{exam:thin:plate:eq}. The point S denotes the center of the source term, and the point M the location where measurements are taken. The top edge is assumed to have a fixed temp}
    \label{fig:pde:examples}
\end{figure}
    We return to the heat transfer problem of \Cref{exam:thin:plate:eq}.
 We represent a 2-dimensional thin plate as the unit square $[0,1]\times[0,1]$. Assume convection coefficients of $\lambda_1,\lambda_2$ in the $x$ and $y$ directions, respectively, and a diffusion coefficient $\lambda_3$. Let \[\phi(x, y, u)=3\exp\left(-\frac{(x-0.5)^2}{0.1}-\frac{(y-u)^2}{0.05}\right)\] denote a heat source underneath the plate, which is positioned at $(0.5,u)$.  We let $T=T(x,y)=T(x, y;\lambda_1,\lambda_2,\lambda_3,u)$ denote the steady state temperature on the square $[0, 1]\times[0,1]$ for a given choice of metal parameters and vertical heat source position, where we  omit the parameters for brevity. Then $T$ solves the partial differential equation  \[-\lambda_3\Delta T + \begin{pmatrix}
     \lambda_1 \\ \lambda_2
 \end{pmatrix}\cdot \nabla T= -\lambda_3\left(\frac{\partial^2 T}{\partial x^2}+\frac{\partial^2 T}{\partial y^2}\right)+\lambda_1\frac{\partial T}{\partial x}+\lambda_2\frac{\partial T}{\partial y}=\phi(x,y,u),\] with Dirichlet boundary conditions,
    \begin{align*}
      T(x, 0) &= 0 \text{ for all }x\in[0, 1],\\
      T(0, y)&=T(1, y)=0 \text{ for all }y\in[0, 1], \\
      T(x, 1) &= \frac{3125}{256} \cdot x(1-x)^4 \text{ for all }x\in[0, 1].
    \end{align*} The bottom, left, and right edges are constrained to 0 while the top edge has a non-homogeneous, localized profile with maximum value 1, representing heat from where the plate is attached. We solve the PDE for each choice of parameters using the FiPy package in Python \cite{FiPy:2009}. We assume the temperature at $(0.5, 0.5)$ is recorded, which establishes the forward model $(\lambda_1,\lambda_2,\lambda_3,u)\mapsto T(0.5, 0.5;\lambda_1, \lambda_2,\lambda_3,u)$. We display in \Cref{fig:pde:examples} some  solutions for different choices of parameters. The parameters impact how the constant heat source at the point S as well as on the top edge spread to the rest of the plate.

\begin{figure}[t]
    \centering  
    \includegraphics[width=0.7\textwidth]{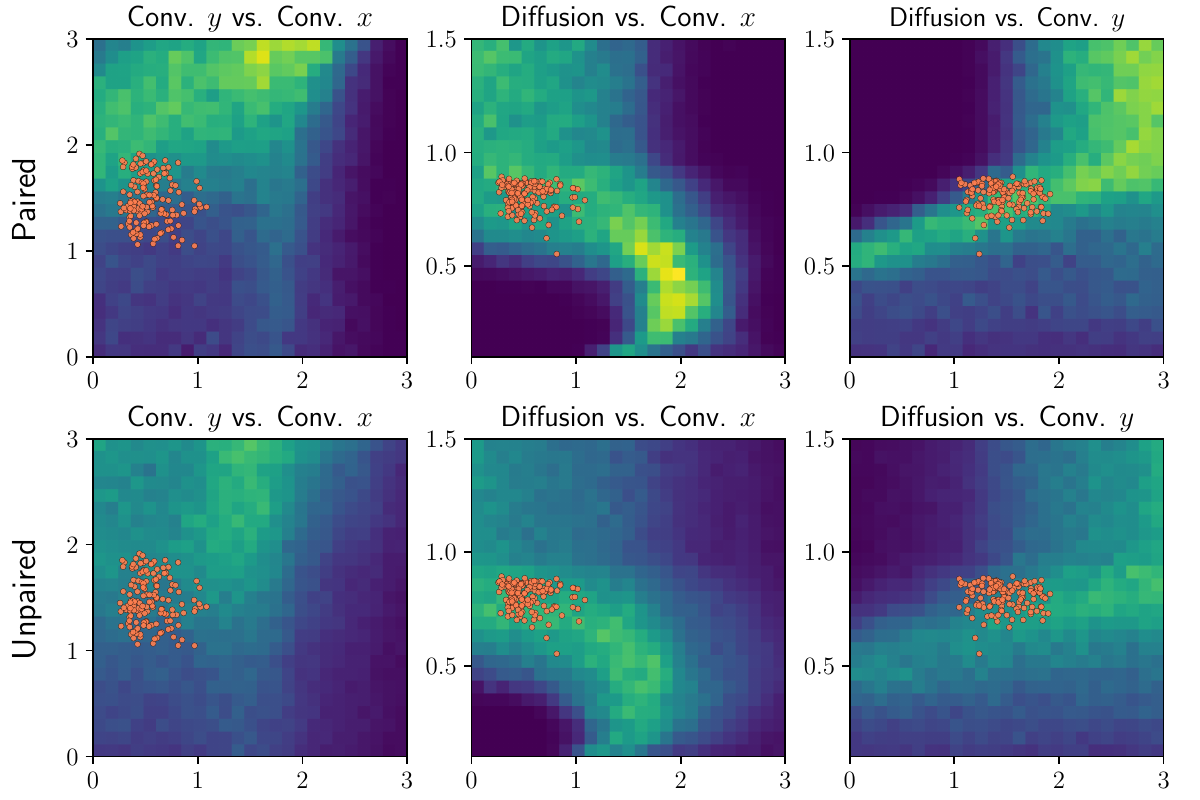}
    \caption{\Cref{exam:thin:plate:eq}'s marginal constrained SCP solutions for each pair of dimensions in $\Lambda$ in both the paired and unpaired case. We refer to $\lambda_1$ and $\lambda_2$ as `Conv. $x$' and `Conv. $y$', respectively. Projections of a sample of TGD points onto each respective subspace are shown. All plots use the same color scale.}
    \label{fig:heat:eq:marginals}
\end{figure}

\begin{figure}[t]
    \centering
    \includegraphics[width=0.7\textwidth]{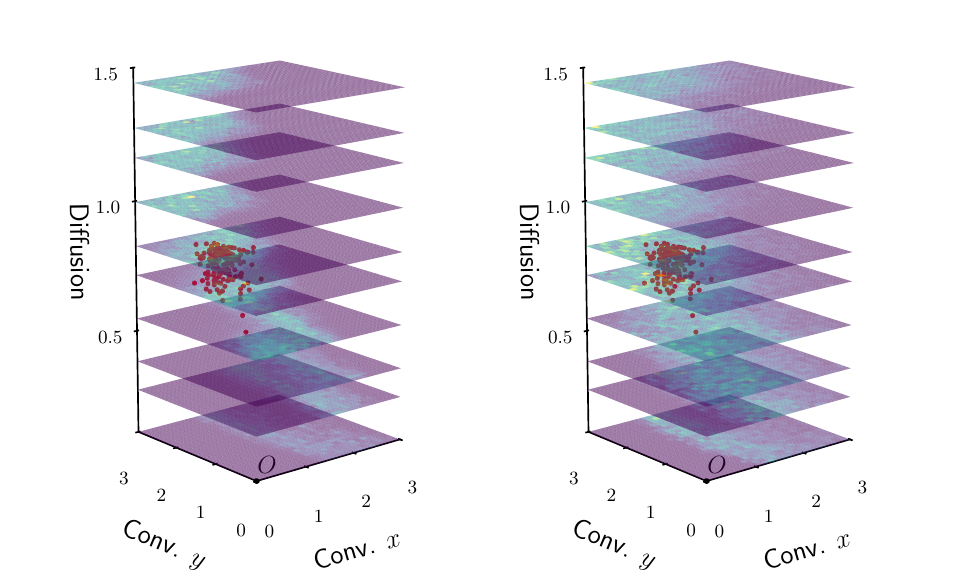}
    \caption{\Cref{exam:thin:plate:eq}'s marginal constrained SCP solutions over the three-dimensional $\Lambda$ in both the paired (left) and unpaired (right) cases. Projected TGD points onto $\Lambda$ are plotted as well.}
    \label{fig:heat:eq:long}
\end{figure}

\begin{figure}[t]
    \centering
    \includegraphics[width=0.7\textwidth]{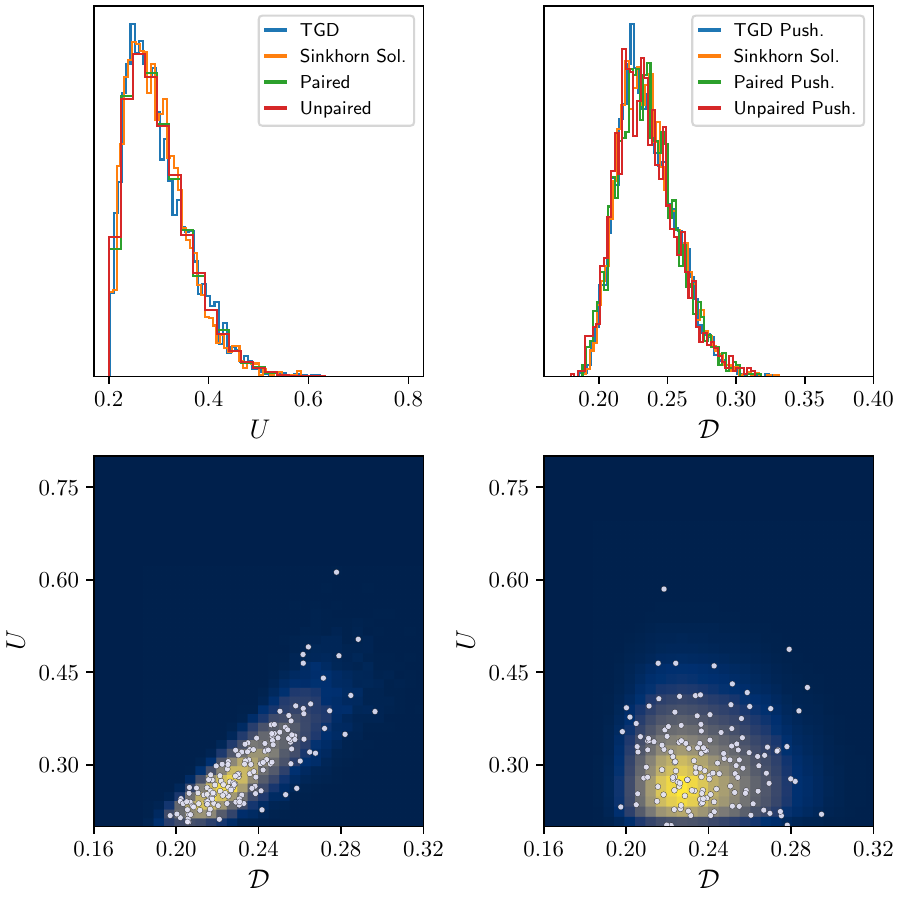}
    \caption{Associated $\cD$ and $U$ plots from \Cref{exam:thin:plate:eq}. Top left shows $U$ marginals of TGD, SK estimate, and both paired and unpaired SCP solutions. Top right shows $\cD$ marginal of SK solution, and the pushforward by $Q$ of the TGD and either SCP solutions. Bottom left shows the pushforward $\tilde Q \TGDjoint$ and the bottom right the SK estimate on $\cD\times U$.}
    \label{fig:heat:eq:diagnostics}
\end{figure}

    We  assume the convection and diffusion parameters are chosen from $\Lambda=[0, 3]\times[0,3] \times [0.1, 1.5]$ and the  vertical position of the source term $\phi$ from $U=[0.2, 0.8]$. For the $i$th-dimension of $\Lambda\times U$, of the form $[a_i, b_i]$, we assume the marginal TGD is a measure with distribution $(b_i-a_i)\Beta(\alpha_i, \beta_i)+a_i$, and we take $(\alpha_1,\beta_1)=(2, 6)$, $(\alpha_2,\beta_2)=(2, 2)$, $(\alpha_3,\beta_3)= (6, 2)$, and $(\alpha_4,\beta_4)=(2, 10)$. For example, the marginal distribution on $U$ is $(0.8-0.2)\cdot \Beta(2, 10)+0.2$, which is a $\Beta(2,10)$ distribution scaled to $[0.2, 0.8]$. The joint TGD is then the independent coupling of these shifted $\Beta$ distributions. It is assumed only that the marginal distribution on $U$ is known for the TGD. The prior distribution is as before the uniform distribution on $\Lambda$ coupled with $\TGDU$. Although \Cref{example:heat:equation:intro} discussed a Gaussian distribution on $U$, we use a $\Beta$ distribution for the simplicity of having a compact support inside $U$.
    
    Rather than directly computing   $Q(\lambda_1,\lambda_2,\lambda_3,u)=T(0.5, 0.5;\lambda_1,\lambda_2,\lambda_3,u)$ to solve the SCP, we train a surrogate model to speed up computations in the SK estimation and importance sampling steps. Importantly, we note that the initial data generation uses the results of the FiPy solver; it is only subsequent evaluations of $Q$ by the algorithm which use the surrogate model. To train the model, we draw a fine grid of points over $\Lambda\times U$, comprised of $\num{50,000}$ points drawn uniformly over this hyper-rectangle, and then $\num{10,000}$ points drawn from $\Beta(0.4, 0.4)$ distributions translated and scaled to each axis of $\Lambda\times U$ (to better sample the edges and corners). We perform a one-time evaluation of $T$ at position $(0.5,0.5)$ for each   choice of parameters and then train a gradient boosted regression model using the XGBoost library in Python \cite{XGBoost}. The map $Q$ is the predicted temperature of this regression model at a given parameter vector $(\lambda_1,\lambda_2,\lambda_3,u)$. The use of a surrogate yields a speed-up of 3--4 orders of magnitude over evaluating the FiPy PDE solver.

    Once $Q$ is trained, we solve the paired case with $K=\num{25,000}$ observations (from the non-surrogate model), $M=20$ bins over $\cD$ and $U$, and $J=\num{300,000}$ prior samples in \Cref{algorithm:constrained:paired}, and for visualization  of the posterior we use $25$ bins per dimension of $\Lambda\times U$. For the unpaired case using \Cref{algorithm:constrained:unpaired}, we set $K_1=K_2=\num{25,000}$, $M_1=M_2=30$, $K_3=\num{15,000}$, $K_4=\num{100,000}$, $\varepsilon=1$,  $J=\num{300,000}$, $M=20$, and $25$ bins per dimension of the input space. 

    In \Cref{fig:heat:eq:marginals}, we display the marginal SCP solutions for each pair of dimensions in $\Lambda$, along with the corresponding projection of latent TGD points.  The top row shows the paired case solutions while the bottom row the unpaired ones. We note that the entire support in each subspace is shown. Outside of the given heatmaps, the estimated probability is 0. We also show in \Cref{fig:heat:eq:long} the marginal SCP solutions over all of $\Lambda$, with just $U$ marginalized out. \Cref{fig:heat:eq:diagnostics} shows the analogous diagnostic checks as \Cref{fig:euclidean:norm:diagnostics}, and indeed, all solutions have the desired pushforward properties (by $Q$ and $\pi_U$).
    
    Generally, the contours in the marginal solution plots are more accentuated in the paired case, and for those involving diffusion, seem to pass through the projected points. This remained consistent for different versions of the experiment (e.g., using a smaller or larger $\Lambda$ or shifting the concentration of the TGD). Some more interesting behavior occurs with the convection-only plots, also consistent across experimental configurations. The main `ridges' seem to miss the projections of the points, if contours could be observed at all. This could be for several reasons. First, the variation in output for different convection parameters may be minimal, making estimation difficult. Second, the marginal of a probability distribution that is actually over 4-dimensions can be misleading alone; even in 2-dimensions, univariate marginals can obscure where the bulk of the mass is. Inspecting \Cref{fig:heat:eq:long}, albeit itself a marginal plot, suggests the points do lie within an oblique contour passing through the region.
    
    In practice, a much more informed prior could be used in such a physical example, using past experimental data about material properties as well as domain knowledge about the type of metal being heated---even if there is still  some stochastic variation. We gave equal weight to a relatively large range of input values, which can correspond to very distinct heat evolutions, as \Cref{fig:pde:examples} demonstrates.

    \end{exam}

\section{Discussion}
\label{section:discussion}

We have extended an important class of non-parametric statistical calibration algorithms for computer models that have stochastic control parameters, i.e., partially specified stochastic input parameters. A scientifically plausible solution should be consistent with the known marginal distribution for the control parameters. Even if a prior with the correct marginal is used, we demonstrate that a straightforward application of SCP solution methodology in \cite{esip_2025} need not preserve the known marginal.

We introduce a method that relies on augmenting the computer model  as well as the output space  by the input space component with a known distribution. We considered two cases, depending on if the field observations are paired with the control parameter data. This amounts to whether the pushforward of the input space by the augmented computer model is known or not---noting that the marginals of this pushforward are known regardless. When the pushforward is known, we could rely on existing SCP techniques with some modifications. In the more challenging case, we construct a substitute for the unknown joint pushforward that preserve the given marginals. 

The method in the unpaired case was based on optimal transport, with a cost functional defined by the computer model. We solve the optimal transport problem numerically using the Sinkhorn-Knopp algorithm. We found the solution is quite robust with respect to the Sinkhorn-Knopp regularization hyperparameter $\varepsilon$. This is surprising since the corresponding estimated pushforward changes substantially. Understanding why this robustness is the case warrants a deeper theoretical investigation. An additional avenue of investigation is the use of multi-marginal formulations of the Sinkhorn-Knopp algorithm, which could pose novel computational or methodological questions.

The heat equation example presents an interesting theoretical issue. Rather than using a PDE as the forward model, we trained a surrogate model using machine learning and used its prediction as the computer model. This meant the PDE computations for the forward model in the SCP solution algorithms only had to be run once, even if the algorithm is re-run with different specifications. In our case, computation was cheap enough that the surrogate model could be trained to essentially perfectly interpolate the computer model mapping. Scenarios with more expensive computer models will not have enough training data for such high interpolation accuracy. Consequently, the surrogate model now has additional statistical error impacting the SCP solutions. Additionally, using black-box surrogate models means potentially violating the regularity conditions of the computer model, and also diverging from physical reality in unpredictable ways. Future work is needed to investigate conditions under which using surrogate models permits the same theoretical guarantees for the SCP.

Lastly, although this paper treats stochastic control and calibration parameters, our methods and theory may yield insights for the deterministic case as well. Recent work established that the unconstrained SCP is continuous with respect to the trial-generating distribution \cite{prasadan2026continuity}. The same results, with some modifications, may also apply in this set-up, which could enable considering the limiting case where the distribution on the parameters converges to Dirac delta measures. This means our stochastic results have implications for classic calibration and parameter estimation problems.

\acknowledgements

A. Prasadan acknowledges the support of the Natural Sciences and Engineering Research Council of Canada.  S. Basu acknowledges the support of the Natural Sciences and Engineering Research Council of Canada. F. Yazdi acknowledges the support of the Natural Sciences and Engineering Research Council of Canada. D. Bingham acknowledges the support of the Natural Sciences and Engineering Research Council of Canada. D. Estep acknowledges the support of the Canada Research Chairs Program and the Natural Sciences and Engineering Research Council of Canada.

\appendix
\crefalias{section}{appendix}

\section{Regularity Conditions for the SCP}
\label{section:regularity:conditions}

The unconstrained SCP solution in \cite{esip_2025} imposes several regularity conditions that imply a well-defined solution, convergence of the importance sampler, a density for the solution, among others desirable features. We restate them below, using the notation and set-up of \Cref{section:unconstrained:SCP}, and then discuss the necessary adaptations for the constrained SCP.

\begin{enumerate}
    \item The set $\Lambda\subseteq\RR^n$ is compact, with measure $0$ boundary. The closure of $\mathrm{int}(\Lambda)$ is $\Lambda$. The set $\cD=Q(\Lambda)\subset\RR^m$ where $m\le n$.
    \item There is some open set $V_{\Lambda}$ containing $\Lambda$ on which $Q$ is a.e. continuously differentiable. 
    \item The  $m\times n$ Jacobian matrix  $J_Q$ of $Q$ has rank $m$ at all points in $V_{\Lambda}$ except on finitely many manifolds of dimension $\le n-1$.
    \item The prior $P_p$ has a density $\rho_p$ with respect to the Lebesgue measure and $Q\TGDmarg\ll QP_p$.
    \item The boundary $\partial \Lambda=\{\lambda: B(\lambda)=0\}$ where $B\colon V_{\Lambda}\to \RR$ is continuously differentiable, and $\partial\Lambda\cap Q^{-1}(q)=\{B(\lambda)=0\}\cap\{Q(\lambda)=q\}$.  Let $J_B$ denote the Jacobian matrix of $B$. If $q\in\cD$ is such that $\partial\Lambda\cap Q^{-1}(q)$ is non-empty, then  $(J_Q, J_B)^T$ has full rank on $\partial\Lambda\cap Q^{-1}(q)$.
    \item The pushforward $Q\TGDmarg$ has a density in Lebesgue measure on $\cD$.
\end{enumerate}

We now consider the setting where the input space is $\Lambda\times U$, $Q\colon\Lambda\times U\to\cD$, the prior measure $P_p$ is placed on $\Lambda\times U$, and the TGD is now $\TGDjoint$ with marginal $\TGDU$. If we solve the fully unconstrained SCP, rather than the proposed techniques, all of the assumptions are swapped out with their equivalent version in the product space. For example, Assumption 1 requires $m\le n+k$, and Assumption 3 that the $m\times(n+k)$ Jacobian matrix $J_Q$ of $Q$ must have full rank.

However, the proposed technique augments $Q$ to $\tilde Q\colon\Lambda\times U\to \cD\times U$ as in \Cref{eq:augmented:Q}, which leads to some subtle changes. For example, we get the condition $m+k\le n+k$, equivalent to $m\le n$, which is more restrictive than $m\le n+k$. Interestingly, the Jacobian conditions become less restrictive in the augmented setting. Recall that requiring $J_Q$ to be of full rank is equivalent to requiring $\mathrm{det}(J_Q J_Q^T)\ne 0$. For the augmented $\tilde Q$, observe that \[J_{\tilde Q}= \begin{bmatrix}
    D_{\lambda} Q & D_u Q \\ 0_{k\times n} & I_k
\end{bmatrix},\] where $D_{\lambda}Q$ and $D_u Q$ are the respective $m\times n$ and $m\times k$ matrices of derivatives over $\lambda$ and $u$. Consequently, an easy calculation shows that $\mathrm{det}(J_{\tilde Q} J_{\tilde Q}^T)\ne 0$ reduces to $\mathrm{det}((D_{\lambda}Q )(D_{\lambda}Q)^T)\ne 0$. In other words, we only have to check the $\lambda$ derivatives of $Q$ in contrast to the non-augmented, unconstrained SCP. The pushforward conditions in Assumptions 4 and 6 are now stated in terms of $\tilde Q$, and the pre-images of Assumption 5 are now of the form $\tilde Q^{-1}(q,u)$.

\section{Alternative Formulation of Optimal Transport Problem}
\label{section:measurability}

As discussed in \Cref{subsection:unpaired:data}, another approach to generating a surrogate of $\tilde Q \TGDjoint$ is to first solve
\[\hat \Pi_{\lambda} = \argmin_{\Gamma\in\mathfrak{F}}\left[\int_{\cD\times U}c_{\lambda}(q,u) d\Gamma(q, u) + \varepsilon\cdot \KL(\Gamma, P_{\cD}\times \TGDU)\right]\] where $c_{\lambda}(q,u)=|q-Q(\lambda, u)|^2$. Then, the final estimate $\hat \Pi$ is the average over the $\{\hat \Pi_{\lambda}\}$, with respect to some prior distribution on $\Lambda$. We must verify, however, that  $\lambda\mapsto \hat\Pi_{\Lambda}$ is a measurable mapping with respect to the appropriate measurable spaces. This turns out to be a simple consequence of a measurability result for nonlinear 
Kantorovich problems with parametrized cost functions due to \cite{Bogachev2022NonlinearKP}.

 Observe that any $\Gamma\in\mathfrak{F}$ satisfies $\Gamma(\cD\times U)=1$ by definition as a joint probability measure on $\cD\times U$. Moreover, the regularization term $\varepsilon\cdot\KL(\Gamma, P_{\cD}\times \TGDU)$ is constant in $q$ and $u$. Thus, we rewrite
\[\hat\Pi_{\lambda}= \argmin_{\Gamma\in\mathfrak{F}} \int_{\cD\times U}h_{\lambda}(q,v,\Gamma) d\Gamma(q,v),\] where \[h_{\lambda}(q,u,\Gamma) = c_{\lambda}(q,v) + \varepsilon\cdot\mathrm{KL}(\Gamma, P_\cD\times \TGDU).\]

 Next, let $\cP_r:=\cP_r(\cD\times U)$ be the set of Radon probability measures on $\cD\times U$. Equip $\cP_r$ with the weak topology induced by bounded and continuous functions. Let $\cB_{\cP_r}$ be the induced Borel $\sigma$-algebra.

\begin{theorem} There is a family of measures $\{\hat\Gamma_{\lambda}\}_{\lambda\in\Lambda}$ with each $\hat\Gamma_{\lambda}\in \hat\Pi_{\lambda}$ such that the mapping $\lambda\mapsto \hat\Gamma_{\lambda}$ is $(\cB_{\Lambda},\cB_{\cP_r})$-measurable.
\end{theorem}
\begin{proof} The spaces $\cD$ and $U$ are complete and separable metric spaces, and $\Lambda$ is a Souslin space.
    The constant maps $\lambda\mapsto P_{\cD}$ and $\lambda\mapsto \TGDU$ are Borel mappings.  For each $\lambda$, note \[\inf_{\Gamma\in\mathfrak{F}} \int_{\cD\times U} h_{\lambda}(q,u,\Gamma) d\Gamma(q,u)<\infty.\] 

    Now pick a compact $K\subset \cD\times U$ and any $\lambda\in\Lambda$. We must show the mapping $(x,y,\Gamma)\mapsto h_{\lambda}(x,y,\Gamma)$ is lower semicontinuous on $K\times \mathfrak{F}(P_{\cD},\TGDU)$.  The  $c_{\lambda}$ term in $h_{\lambda}$ is clearly smooth in $(x,y)$ and constant in $\Gamma$, and the mapping $\Gamma\mapsto \KL(\Gamma, P_\cD\times \TGDU)$ is lower-semicontinuous by \cite[Theorem 1]{posner}. Therefore, their sum is lower-semicontinuous. The claimed family of measures $\{\hat\Gamma_{\lambda}\}
    _{\lambda\in\Lambda}$ then exists by \cite[Theorem 1]{Bogachev2022NonlinearKP}.
\end{proof}

\bibliographystyle{IJ4UQ_Bibliography_Style}

\bibliography{bibliography}

\end{document}